\documentclass[11pt, a4paper]{article}
\usepackage[top=0.94in, bottom=1in, left=1in, right=1in]{geometry}
\usepackage{algorithmic}
\usepackage{algorithm}
\usepackage{amssymb}
\usepackage{amsmath}
\usepackage{amsthm}

\newtheoremstyle{willthm}
     {4pt}
     {4pt}
     {}
     {}
     {\bfseries}
     {.}
     { }
     {}

\theoremstyle{willthm}

\newtheorem{thm}{Theorem}[section]
\newtheorem{lem}[thm]{Lemma}
\newtheorem{cor}[thm]{Corollary}

\newtheorem{rem}[thm]{Remark}

\newtheorem{dfn}[thm]{Definition}

\newcommand{\abs}[1]{\left|#1\right|}
\newcommand{\calA}{\mathcal{A}}
\newcommand{\calP}{\mathcal{P}}
\newcommand{\ceil}[1]{\left\lceil#1\right\rceil}
\newcommand{\commgraph}{G}
\newcommand{\dft}[1]{\textbf{\textit{#1}}}
\newcommand{\e}{\varepsilon}
\newcommand{\matching}{M}

\newcommand{\paren}[1]{\left(#1\right)}
\newcommand{\set}[1]{\left\{#1\right\}}
\newcommand{\st}{\middle |}

\DeclareMathOperator{\E}{\mathbf{E}}

\def\AMM{{\bfseries AMM}}

\def\ASM{{\bfseries ASM}}
\def\MatchingRound{{\bfseries MatchingRound}}
\def\MaximalMatching{{\bfseries MaximalMatching}}
\def\ProposalRound{{\bfseries ProposalRound}}
\def\QuantileMatch{{\bfseries QuantileMatch}}
\def\RandASM{{\bfseries RandASM}}
\def\AlmostRegularASM{{\bfseries AlmostRegularASM}}

\title{Fast distributed almost stable matchings}
\author{
Rafail Ostrovsky\thanks{University of California, Los Angeles (Departments of Computer Science and Mathematics). Work supported in part by NSF grants 09165174, 1065276, 1118126 and 1136174, US-Israel BSF grant 2008411, OKAWA Foundation Research Award, IBM Faculty Research Award, Xerox Faculty Research Award, B. John Garrick Foundation Award, Teradata Research Award, and Lockheed-Martin Corporation Research Award. This material is based upon work supported by the Defense Advanced Research Projects Agency through the U.S. Office of Naval Research under Contract N00014-11-1-0392. The views expressed are those of the author and do not reflect the official policy or position of the Department of Defense or the U.S. Government.}
\and
Will Rosenbaum\thanks{University of California, Los Angeles (Department of Mathematics).}
}

\date{\today}

\begin{document}
  \maketitle
  \thispagestyle{empty}
  \begin{abstract}
    \normalsize
    In their seminal work on the Stable Marriage Problem, Gale and Shapley \cite{GS62} describe an algorithm which finds a stable matching in $O(n^2)$ communication rounds. Their algorithm has a natural interpretation as a distributed algorithm where each player is represented by a single processor. In this distributed model, Flor{\'e}en, Kaski, Polishchuk, and Suomela \cite{FKPS10} recently showed that for bounded preference lists, terminating the Gale-Shapley algorithm after a constant number of rounds results in an \emph{almost} stable matching. In this paper, we describe a new deterministic distributed algorithm which finds an almost stable matching in $O(\log^5 n)$ communication rounds for \emph{arbitrary} preferences. We also present a faster randomized variant which requires $O(\log^2 n)$ rounds. This run-time can be improved to $O(1)$ rounds for ``almost regular'' (and in particular complete) preferences. To our knowledge, these are the first sub-polynomial round distributed algorithms for any variant of the stable marriage problem with unbounded preferences.
  \end{abstract}

  \section{Introduction}
  \label{sec:intro}

  \subsection{Historical Background}


    In their seminal work, Gale and Shapley \cite{GS62} consider the following problem. Members of disjoint sets of $n$ \dft{men} and $n$ \dft{women} each rank all members of the opposite sex. The men and women (which we collectively call \dft{players}) wish to form a \dft{matching}---a one-to-one correspondence between the men and women---which is \dft{stable} in the sense that it contains no \dft{blocking pairs}: pairs of players who mutually prefer each other to their assigned partners in the matching. Gale and Shapley showed that a stable matching always exists by giving an explicit algorithm for finding one. The centralized Gale-Shapley algorithm runs in time $\tilde{O}(n^2)$, and this run-time is asymptotically optimal for centralized algorithms \cite{GI89}. The Gale-Shapley algorithm easily generalizes to the case of \dft{incomplete preferences} where each player ranks only a subset of the members of the opposite sex \cite{GI89}.

  The Gale-Shapley algorithm has a natural interpretation as a distributed algorithm, where each player is represented by a separate processor which privately holds that player's preferences. The communication links between players are formed by pairs of players who appear on each other's preference lists. This model is natural in, for example, social networks where players may be constrained to be matched with acquaintances and do not communicate with strangers. In this model, the input to each processor has size $\tilde{O}(n)$, yet there is still no known distributed algorithm which improves upon the Gale-Shapley algorithm's $\tilde{O}(n^2)$ run-time for arbitrary preferences.\footnote{In the distributed computational model with complete preferences, each player can broadcast their preferences to all other players in $O(n)$ rounds, after which each player runs a centralized version of the Gale-Shapley algorithm. While this process requires only $O(n)$ communication \emph{rounds}, the synchronous distributed run-\emph{time} is still $\tilde{\Theta}(n^2)$ in the worst case.} 


  Recently, there has been interest in approximate versions of the stable marriage problem \cite{ABM06, EH08, FKPS10, HMV14, KP09}, where the goal is to find a matching which is ``almost stable.'' There is no consensus in the literature on precisely how to measure almost stability, but typically almost stability requires that a matching induces relatively few blocking pairs. Eriksson and H{\"a}ggstr{\"o}m \cite{EH08} argue that, ``the proportion of blocking pairs among all possible pairs is usually the best measure of instability.'' Using a finer notion of almost stability, Flor{\'e}en, Kaski, Polishchuk, and Suomela show \cite{FKPS10} that for \emph{bounded} preference lists, truncating the Gale-Shapley algorithm after boundedly many communication rounds yields a matching that induces at most $\e \abs{M}$ blocking pairs. Here $\abs{M}$ is the size of the matching produced. More recently, Hassidim, Mansour and Vardi \cite{HMV14} show a similar result in a more restrictive ``local'' computational model, so long as the men's preferences are chosen uniformly at random.


    Kipnis and Patt-Shamir \cite{KP09} give an algorithm which finds an almost stable matching using $O(n)$ communication rounds in the worst case, using a finer notion of approximate stability than we consider. Specifically, in their notion of almost-stability, a matching is almost stable if no pair of players can both improve their match by more than an $\e$-fraction of their preference list by deviating from their assigned partners. They also prove an $\Omega(\sqrt{n} / \log n)$ communication round lower bound for finding an approximate stable matching for this notion of approximation.

  \subsection{Overview of Results}

  We consider an approximate version of the stable marriage problem where an almost stable matching is allowed to have $\e \abs{E}$ blocking pairs. Here $E$ is the set of pairs of men and women who rank one another (that is, the set of edges in the communication graph). Our notion of approximation, which generalizes almost stability as  described in \cite{EH08}, is strictly coarser than those used in \cite{FKPS10} and \cite{KP09}. However, for bounded preference (the context of \cite{FKPS10}), our notion of instability agrees with that of \cite{FKPS10} up to a constant factor.

  Using the notion of almost stability given above, we describe a deterministic distributed algorithm, \ASM, which produces an almost stable matching in $O(\log^5(n))$ rounds. We note that in order to obtain this sub-polynomial run-time, we cannot use a finer notion of approximation than \cite{KP09}, who prove an $\Omega(\sqrt{n} / \log n)$ round lower bound for their model. We remark that the after removing an arbitrarily small fraction of ``bad'' players, the output of \ASM\ is almost stable in the sense of \cite{KP09} as well. We further describe a faster randomized variant of \ASM\ which runs in $O(\log^2(n))$ rounds. For preferences which are ``almost regular,'' (and in particular for complete preferences) this run-time can be improved to $O(1)$.

  \begin{thm}
    \label{thm:main}
    There exists a deterministic distributed algorithm \ASM\ which produces a $(1 - \e)$-stable matching in $O(\log^5(n))$ communication rounds. A randomized variant of the algorithm, \RandASM\ runs in $O(\log^2(n))$ rounds for general preferences, and can be improved to $O(1)$ rounds for almost regular (and in particular complete) preferences.
  \end{thm}


  \ASM\ can be viewed as a generalization of the classical Gale-Shapley algorithm \cite{GS62} which allows for multiple simultaneous proposals by the men and acceptances by the women. In \ASM, the players quantize their preferences into $O(\e^{-1})$ quantiles of equal size. In each step of the algorithm, the men propose to all women in their best nonempty quantile. Each woman accepts proposals only from her best quantile receiving proposals. A maximal matching is then found among the accepted proposals, and matched women reject men they do not prefer to their matches. This procedure is iterated until a large fraction of men are either matched or have been rejected by all women. 

  The analysis of our algorithm follows in two steps. We first show that by quantizing preferences, the matching found by \ASM\ cannot contain a large fraction of blocking pairs among the matched (or rejected) players. We bound the number of blocking pairs from the remaining ``bad'' players by showing there are few such players, and that only a small fraction can participate in many blocking pairs.

  The remainder of the paper is organized as follows. In Section \ref{sec:preliminaries} we formalize our notion almost stable matchings and our computational model. We also overview methods of computing maximal matchings which our algorithm will require as subroutines. Section \ref{sec:description} describes \ASM\ and its subroutines in detail and states basic guarantees for the subroutines. Section \ref{sec:performance} proves the performance guarantees for \ASM. Finally, in Section \ref{sec:randomized} we describe the randomized variants of \ASM.


  \section{Preliminaries}
  \label{sec:preliminaries}

  \subsection{Stable and almost stable matchings}


  We consider the stable marriage problem as originally described by Gale and Shapley \cite{GS62} with incomplete preferences or, equivalently, unacceptable partners (cf. \cite{GI89, Manlove13}). Let $X$ and $Y$ be sets of women and men, respectively. For simplicity, we assume $\abs{X} = \abs{Y} = n$. Each player $v \in X \cup Y$ holds a \dft{preference list} or \dft{ranking} $P^v$---a linear order on a subset of the members of the opposite sex. We denote the set of all player's preferences by $\calP = \set{P^v \st v \in X \cup Y}$.  We refer to the players $u$ that appear on $v$'s preference list $P^v$ as $v$'s \dft{acceptable partners}. We call $\calP$ \dft{complete} if each player ranks all players of the opposite sex. If a man $m$ precedes $m'$ on woman $w$'s preference list, we write $m \succ_w m'$, and we say that $w$ \dft{prefers} $m$ to $m'$. For simplicity, we assume that preferences are symmetric in the sense that if $m$ appears in $P^w$, then $w$ appears in $P^m$. Given players $v$ and $u$ of opposite genders, we let $P^v(u)$ denote $v$'s \dft{rank} of $u$. For example, $P^v(u) = 1$ means that $u$ is $v$'s most favored partner, et cetera.

  We define the \dft{communication graph} $\commgraph = (V, E)$ for a set of preferences $\calP$ to be
  \[
  V = Y \times X, \qquad E = \set{(m, w) \st m \in P^w,\ w \in P^m}.
  \]
  For a communication graph $\commgraph = (V, E)$, we denote the degree of $v \in V$ by $\deg v$, which is the number of players that appear on $v$'s preference list.
  
  A \dft{matching} $\matching \subseteq E$ is a set of edges in $E$ such that no two edges share a vertex. Given a matching $\matching$ and $(m, w) \in \matching$, we call $m$ and $w$ \dft{partners} and write $p(w) = m$ and $p(m) = w$. Given preferences and a matching $\matching$, we say that an edge $(m, w) \in E$ is a \dft{blocking pair} if $(m, w) \notin \matching$, but $m$ and $w$ mutually prefer each other to their partners in $\matching$; that is,
  \[
  m \succ_w p(w) \quad\text{and}\quad w \succ_m p(m).
  \]
  By convention, we assume each unmatched player ($p(v) = \varnothing$) prefers all acceptable partners to being without a partner. A \dft{stable matching} is a matching which contains no blocking pairs.

  We are primarily concerned with finding matchings which are ``almost stable'' in the sense that they induce relatively few blocking pairs. We use a definition of almost stability given by Eriksson and H{\"a}ggstr{\"o}m \cite{EH08}, modified to allow for incomplete preference lists.
  \begin{dfn}
    \label{dfn:as}
    Given $\e \geq 0$ and preferences $\calP$, we say that a matching $\matching$ is \dft{$(1 - \e)$-stable} with respect to $\calP$ if $\matching$ induces at most $\e \abs{E}$ blocking pairs with respect to $\calP$.
  \end{dfn}

  We refer to the problem of finding a matching which is $(1 - \e)$-stable for fixed $\e > 0$ as the \dft{almost stable matching problem}. Note that a $1$-stable matching corresponds precisely to the classical stable matching definition.

  \begin{rem}
    \label{rem:fkps}
    Again, we reiterate that there is no consensus in the literature on the precise definition of almost stability. For example, the authors of \cite{FKPS10} compare the number of blocking pairs to $\abs{\matching}$, the size of the matching rather than $\abs{E}$, as we do. Since in \cite{FKPS10}, they only consider bounded preference lists, their notion of almost stability agrees with Definition \ref{dfn:as} up to a constant factor.
  \end{rem}

  \begin{dfn}[Kipnis and Patt-Shamir \cite{KP09}]
    \label{dfn:eps-blocking}
    Given $\e \geq 0$, preferences $\calP$, and a matching $\matching$ we call an edge $(m, w) \in E$ \dft{$\e$-blocking} if $m$ and $w$ appear an $\e$-fraction higher on each other's preferences than their assigned partners. Specifically, $(m, w)$ is $\e$-blocking if
    \[
    P^m(p(m)) - P^m(w) \geq \e \deg m \quad\text{and}\quad P^w(p(w)) - P^w(m) \geq \e \deg w.
    \]
    We say that $\matching$ is \dft{$\e$-blocking-stable} if it contains no $\e$-blocking pairs.
  \end{dfn}

  \begin{rem}
    \label{rem:kp}
    Kipnis and Patt-Shamir \cite{KP09} prove an $\Omega(\sqrt{n} / \log n)$ round lower bound for finding an $\e$-blocking-stable matching. That we are able to achieve a polylogarithmic round algorithm for the almost stable matching problem using Definition \ref{dfn:as} bolsters the use of Definition \ref{dfn:as} for almost stability, at least for practical applications. Further, \ASM\ produces a matching which is nearly $\e$-blocking-stable in the sense that after the removal of an arbitrarily small fraction of ``bad'' men, the resulting matching is $\e$-blocking stable with respect to the remaining players.
  \end{rem}

  \subsection{Computational model}


    We describe \ASM\ in terms of the CONGEST model formalized by Peleg \cite{Peleg00}. In this distributed computational model, each player $v \in X \cup Y$ represents a processor. Given preferences $\calP$, the communication links between the players are given by the set of edges $E$ in the communication graph $\commgraph$. Communication is performed in synchronous rounds. Each communication round occurs in three stages: first, each processor receives messages (if any) sent from its neighbors in $\commgraph$ during the previous round. Next, each processor performs local calculations based on its internal state and any received messages. We make no restrictions on the complexity of local computations. Finally, each processor sends short ($O(\log n)$ bit) messages to its neighbors in $\commgraph$---the processor may send distinct messages to distinct neighbors. In the CONGEST model, complexity is measured by the number of communication rounds needed to solve a problem.

    \begin{rem}
      Although the CONGEST model allows for unbounded local computation during each round, the computations required by \ASM\ can be implemented in linear or near-linear time in each processor's input.
    \end{rem}

  \subsection{Maximal matchings}
  
  As a subroutine, \ASM\ requires a method for computing maximal matchings in a graph. 

  \begin{dfn}
    \label{dfn:maximal-matching}
    A matching $\matching$ is a \dft{maximal matching} if it is not properly contained in any larger matching. Equivalently, $\matching$ is maximal if and only if every $v \in V$ satisfies precisely one of the following conditions:
    \begin{enumerate}
    \item there exists a unique $u \in V$ with $(v, u) \in \matching$;
    \item for all $u \in N(v)$ there exists $v' \in V$ with $v' \neq v$ such that $(v', u) \in \matching$.
    \end{enumerate}
  \end{dfn}

  For the deterministic version of our algorithm, we invoke the work of Han\'n\'ckowiak, Karo\'nski, and Panconesi \cite{HKP01} who give a deterministic distributed algorithm which finds a maximal matching in a polylogarithmic number of rounds. 
  \begin{thm}[Han\'n\'ckowiak, Karo\'nski, and Panconesi \cite{HKP01}]
    \label{thm:max-matching}
    There exists a deterministic distributed algorithm, \MaximalMatching, which finds a maximal matching in a communication graph $G = (V, E)$ in $\log^4(n)$ rounds, where $n = \abs{V}$.
  \end{thm}

  We remark that while the authors of \cite{HKP01} do not explicitly use the CONGEST model for computation, their algorithm can easily be implemented in this model.

  The randomized variants of \ASM\ require faster (randomized) subroutines for computing maximal and ``almost maximal'' matchings in a communication graph. In Appendix \ref{sec:amm-proof}, we describe how to modify an algorithm of Israeli and Itai \cite{II86} to give the necessary results.

  \section{Deterministic algorithm description}
  \label{sec:description}

  In this section, we describe in detail the almost stable matching algorithm, \ASM. The main algorithm invokes the subroutine \QuantileMatch\ which in turn calls \ProposalRound. In Section \ref{sec:notation} we introduce notation, and describe the internal state of each processor during the execution of \ASM. Section \ref{sec:proposal-round} contains a description of the \ProposalRound\ subroutine, while Section \ref{sec:quantile-match} describes the \QuantileMatch\ subroutine. Finally, Section \ref{sec:asm} describes \ASM.

  \subsection{The state of a processor}
  \label{sec:notation}

  In our algorithm, we assume that each player is represented by an independent processor. Each processor has a unique id and a gender (male or female) both of which are known to that processor. The only global information known to each processor is (an upper bound on) $n$, the total number of processors. At each step of the algorithm, we specify the state of each processor as well as any messages the processor might send or receive. The state of a player $v$ consists of:
  \begin{itemize}
  \item Quantized preferences $Q_1, Q_2, \ldots, Q_k$ where we denote $Q = \bigcup Q_i$. Initially  $Q_1$ is the set of $v$'s $\deg v / k$ favorite men, $Q_2$ is her next favorite $\deg(v) / k$, and so on. We call $Q_i$ $w$'s \dft{$i$th quantile}. For $m \in Q_i$, we write $q(m) = i$. If we wish to make explicit the player to whom the preferences belong, we may adorn these symbols with a superscript. For example, $Q_i^v$ is $v$'s $i$th quantile. Throughout the execution of the algorithm, elements may be removed from $Q$ and the $Q_i$s, but elements will never be added to any of these sets. We will always use $k$ to denote the number of quantiles for the players' preferences.
  \item A partner $p$ (possibly empty). The partner $p$ is $v$'s current partner in the matching $M$ our algorithm constructs. To emphasize that $p$ is player $v$'s partner, we will write $p(v)$. The (partial) matching $M$ produced by the algorithm at any step is given by $M = \set{(p(w), w) \st w \in X, p(w) \neq \varnothing}$.
  \end{itemize}
  Additionally, subroutines of our algorithm will require each processor to store the following variables:
  \begin{itemize}
  \item A set $\commgraph_0$ of ``neighbors'' of the opposite sex which correspond to accepted proposals.
  \item A partner $p_0$ in a matching found in the graph determined by $\commgraph_0$.
  \end{itemize}
  Thus each player knows their preferences, partners (if any) as well as any of their accepted proposals from the current round (stored in $\commgraph_0$). The men $m \in Y$ hold the following additional information:
  \begin{itemize}
  \item A set $A$ of ``active'' potential mates, initially set to $Q_1$.
  \end{itemize}

  \subsection{The \ProposalRound\ subroutine}
  \label{sec:proposal-round}

  At the heart of our algorithm is the \ProposalRound\ subroutine (Algorithm \ref{alg:proposal-round}). \ProposalRound\ works in 5 steps which are described in Algorithm \ref{alg:proposal-round}.

  \begin{algorithm}
    \caption{\ProposalRound$(Q, k, A)$}
    \label{alg:proposal-round}
    \begin{description}
    \item[Step 1:] Each man $m$ proposes to all women in $A^m$ by sending each $w \in A$ the message PROPOSE.
    \item[Step 2:] Each women $w$ receiving proposals responds with the message ACCEPT to all proposals from her most preferred quantile $Q_i^w$ from which at least one man proposed in Step 1.
    \item[Step 3:] Let $\commgraph_0$ denote the bipartite graph $\commgraph_0$ of accepted proposals from Step 2.  The players compute a maximal matching $M_0$ in $\commgraph_0$, using \MaximalMatching$(\commgraph_0)$, storing their match in $\commgraph_0$ as $p_0$. 
    \item[Step 4:] Each woman $w$ matched in $\matching_0$ sends REJECT to all men $m \in Q^w$ in a lesser or equal quantile to her partner $p_0(w)$ in $M_0$ other than $p_0(w)$. She then removes all of these men from $Q^w$ and the corresponding $Q^w_i$. The matched women then set $p \leftarrow p_0$, so the partial matching $\matching$ now contains the edge $(p_0(w), w)$. Any man $m$ matched in $\matching_0$ sets $p \leftarrow p_0$ and sets $A \leftarrow \varnothing$. 
    \item[Step 5:] The men remove all $w$ from whom they received the message REJECT from their preferences $Q$, the various $Q_i$ and $A$. If a man $m$ receives a rejection from his match $p(m)$ from a previous round, he sets $p \leftarrow \varnothing$.
    \end{description}
  \end{algorithm}

  We observe that if each player $v$ takes $k = \deg v$, then \ProposalRound\ mimics the classical (extended) Gale-Shapley algorithm \cite{GS62, GI89}. In this case, each man proposes to his most favored woman that has not yet rejected him, and each woman rejects all but her most favored suitor. Computing a maximal matching is trivial, as the accepted proposals already form a matching. The general case has one crucial feature in common with the Gale-Shapley algorithm, which follows immediately from the description of \ProposalRound.

  \begin{lem}[Monotonicity]
    \label{lem:w-monotonic}
    Once a woman $w$ has $p(w) \neq \varnothing$ in some execution of \ProposalRound, she is guaranteed to always have $p(w) \neq \varnothing$ after each subsequent execution of \ProposalRound. Further, once matched, she will only accept proposals from men in a strictly higher quantile than her current match, $p(w)$.
  \end{lem}

  \subsection{The \QuantileMatch\ subroutine}
  \label{sec:quantile-match}

  Here we describe the \QuantileMatch\ subroutine (Algorithm \ref{alg:quantile-match}), which simply iterates \ProposalRound\ until each man $m$ has either been rejected by all women in $A^m$ or is matched with some woman in $A^m$. In either case, $A^m = \varnothing$ when \QuantileMatch\ terminates. We will argue that $k$ (the number of quantiles) iterations suffice.

  \begin{algorithm}
    \caption{\QuantileMatch$(Q, k)$}
    \label{alg:quantile-match}
    \begin{algorithmic}
      \STATE $i \leftarrow \min \set{i \st Q_i \neq \varnothing}\cup\set{k}$ (male only)
      \IF{$p = \varnothing$}
      \STATE $A \leftarrow Q_i$ (male only)
      \ENDIF
      \FOR{$i \leftarrow 1$ to $k$}
      \STATE \ProposalRound$\paren{Q, k, A}$
      \ENDFOR
    \end{algorithmic}
  \end{algorithm}    

  \begin{lem}[\QuantileMatch\ guarantee]
    \label{lem:quantile-match}
    At the termination of \QuantileMatch$(Q,k)$ every man $m$ satisfies $A^m = \varnothing$. In particular, each man who had $A^m \neq \varnothing$ before the first iteration of the loop in \QuantileMatch\ has either been rejected by all women in $A^m$ or is matched with some woman in $A^m$. 
  \end{lem}
  \begin{proof}
    Suppose a woman $w$ receives proposals in the first iteration of the loop in \QuantileMatch. If she is matched with one of these suitors when \ProposalRound\ terminates, she rejects all other men and receives no further proposals during the current \QuantileMatch. On the other hand, if she is not matched with one of these suitors after the first round, then by the maximality of the matching found in Step 3 of \ProposalRound, all of the suitors in her best quantile receiving proposals are matched with other women. Thus, in the next iteration, she only receives proposals from men in strictly worse quantiles than she accepted in the first. Similarly, in each iteration of the loop, her best quantile receiving proposals (if any) is strictly worse than the previous iteration. Therefore, after $k$ iterations, no woman will receive proposals, hence each man $m$ must have $A^m = \varnothing$.
  \end{proof}

  \subsection{The \ASM\ algorithm}
  \label{sec:asm}

  In this section, we describe the main algorithm \ASM\ (Algorithm \ref{alg:asm}). The idea of \ASM\ is to iterate \QuantileMatch\ until a large fraction men with high degree are either matched or have been rejected by all acceptable partners. We call such men \dft{good}. By iterating \QuantileMatch\ a constant number of times, we can ensure that the fraction of good men is close to $1$. In order to bound the number of blocking pairs from men which are \dft{bad} (not good), we must ensure that bad men comprise only a small fraction of players with relatively high degree. To this end, we only allow men who are potentially involved in many blocking pairs (that is, with $\abs{Q}$ relatively large) to participate in later calls to \QuantileMatch. 

  \begin{algorithm}
   \caption{\ASM$(P, \e, n)$}
   \label{alg:asm}
    \begin{algorithmic}
      \STATE $k \leftarrow \ceil{8 \e^{-1}},\ \delta \leftarrow \e / 8$
      \FORALL{$i \leq k$}
      \STATE $Q_i \leftarrow \set{v \st q(v) = i}$
      \ENDFOR
      \STATE $Q \leftarrow \bigcup_i Q_i,\ p \leftarrow \varnothing$
      \FOR{$i \leftarrow 0$ to $\log n$}
      \IF{$\abs{Q} \geq 2^i$}
      \FOR{$j \leftarrow 1$ to $2 \delta^{-1} k$}
      \STATE \QuantileMatch$(Q, k)$
      \ENDFOR
      \ENDIF
      \ENDFOR
    \end{algorithmic}
  \end{algorithm}

  \section{Performance guarantees}
  \label{sec:performance}

  Here we analyze the performance of \ASM\ and its subroutines. The run-time guarantee (Theorem \ref{thm:runtime}) is a simple consequence of the description of \ASM\ and its subroutines. To prove the approximation guarantee (Theorem \ref{thm:approximation-guarantee}), we consider blocking edges from two sets of men separately. We call a man $m$ \dft{good} if when \ASM\ terminates, he is either matched or has been rejected by all of his acceptable partners. A man who is not good is \dft{bad}. We denote the sets of good and bad men by $G$ and $B$, respectively. 

  \begin{thm}[Approximation guarantee]
    \label{thm:approximation-guarantee}
    The matching $\matching$ output by \ASM\ induces at most $\e \abs{E}$ blocking pairs with respect to $\calP$. Thus $M$ is $(1 - \e)$-stable.
  \end{thm}


  \subsection{Bounding blocking pairs from good players}
  \label{sec:good}

  We bound the number of blocking pairs from good men in two steps. First we show that the good men are not involved in any $(2/k)$-blocking pairs (see Definition \ref{dfn:eps-blocking}). Next, we show that as a result, the good men can only be incident with a small fraction of blocking pairs.

  \begin{lem}[$(2/k)$-blocking-stability of good men]
    \label{lem:eps-blocking-good}
    Let $m \in G$ be good. Then $m$ is not incident with any $(2/k)$-blocking pairs.
  \end{lem}
  \begin{proof}
    Suppose $m \in G$ and that $(m, w)$ is $(2/k)$-blocking. First consider the case where $m$ is matched, $p(m) \neq \varnothing$. Since $m$'s preferences are divided into $k$ quantiles, $w$ must be in a strictly better quantile than $p(m)$. Thus, $m$ must have proposed to $w$ in a strictly earlier call to \QuantileMatch\ than the call in which he was matched with $p(m)$. Thus, by Lemma \ref{lem:quantile-match}, $m$ must have been rejected by $w$, implying that $w$ was matched with a man $m'$ in the same or better quantile than $m$ in this round. By Lemma \ref{lem:w-monotonic}, $w$'s partner when \ASM\ terminates is at least as desirable as $m'$. This contradicts that $(m, w)$ is $\e$-blocking.

    On the other hand, if $p(m) = \varnothing$, then since $m$ is good, he must have been rejected by all of his acceptable partners, and in particular, by $w$. Thus, as in the previous paragraph, $w$ must be matched with a man in the same or better quantile than $m$. 
  \end{proof}

  \begin{lem}[Few non-$(2/k)$-blocking pairs]
    \label{lem:good-blocking}
    There are at most $4 \abs{E} / k$ blocking pairs which are not $(2/k)$-blocking.
  \end{lem}
  \begin{proof}
    Suppose $(m, w)$ is a blocking pair which is not $(2/k)$-blocking. Thus, we have
    \begin{equation}
      \label{eqn:eps-blocking}
      P^m(w) - P^m(p(m)) \leq 2 \deg(m) / k \quad\text{or}\quad P^w(m) - P^w(p(w)) \leq 2 \deg(w) / k,
    \end{equation}
    where by convention we take $P^m(\varnothing) = \deg(m) + 1$. Let $E_N$ denote the set of blocking pairs which are not $(2/k)$-blocking. For each $m$, the number of edges satisfying the first inequality in (\ref{eqn:eps-blocking}) is at most $2 \deg(m) / k$, and similarly for the women. Thus 
    \[
    \abs{E_N} \leq \sum_{m \in Y} 2 \deg(m) / k + \sum_{w \in X} 2 \deg(w) / k = 4 \abs{E} / k,
    \]
as desired.
  \end{proof}

    Lemma \ref{lem:good-blocking} shows that no good player is involved in any $(2/k)$-blocking pairs. Combining Lemmas \ref{lem:eps-blocking-good} and \ref{lem:good-blocking}, we can bound the number of blocking pairs incident with good men. All that remains is to bound the number of $(2/k)$-blocking pairs incident with bad men. In the next section, we show that the proportion of bad men is small (at most $\delta n$), and bound the number of $(2/k)$-blocking pairs they contribute. We remark that by Lemma \ref{lem:good-blocking} and the lower bound of Kipnis and Patt-Shamir \cite{KP09}, we cannot hope to have all men be good in $o(\sqrt{n} / \log n)$ rounds.

  \subsection{Bounding blocking pairs from bad players}
  \label{sec:bad}

  In this section, we prove the following bound on the number of blocking pairs contributed by the bad men at the termination of \ASM. Throughout the section, for simplicity of notation, we assume that $\log n$ is an integer.

  \begin{lem}[Bad men guarantee]
    \label{lem:bad-guarantee}
    At the termination of \ASM, for any $\delta \leq \frac 1 2$ the bad men contribute at most $4 \delta \abs{E}$  $(2/k)$-blocking pairs.
  \end{lem}

  The proof of Lemma \ref{lem:bad-guarantee} is in two parts corresponding to guarantees for each of the two nested loops in \ASM. We refer to men $m$ with $\abs{Q^m} \geq 2^i$ as \dft{active in the $i$th iteration} of the outer loop; the remaining men are \dft{inactive} in the $i$th iteration. 

  \begin{lem}[Few bad men]
    \label{lem:few-bad-men}
    When the inner loop in \ASM\ terminates, at most a $\delta$-fraction of active men are bad.
  \end{lem}
  \begin{proof}
    Let $\calA$ denote the set of active men before executing the inner loop in \ASM. Suppose that after $\ell$ iterations of the inner loop, there are $b$ bad men in $\calA$. We claim that there must have been at least $b$ bad players in every iteration of the inner loop. To see this, first note that by Lemma \ref{lem:w-monotonic}, the number of matched players (and hence matched men) can only increase with each call to \ProposalRound. Second, if a man is rejected by all women on his preference list, he will never become bad. Therefore, the number of good players can only increase with each iteration of the inner loop. Thus there must have been at least $b$ bad men after each of the $\ell$ iterations of the inner loop.


    Suppose $m$ was bad before some call to \QuantileMatch, so that $A^m \neq \varnothing$. By Lemma \ref{lem:quantile-match}, after \QuantileMatch\ $m$ is either matched, or has been rejected by all women $w \in A^m$. In the former case, $p(m)$ rejected all men in her quantile containing $m$. In either case, $m$ witnessed the rejection of a quantile of men---either by precipitating the rejection of $p(m)$'s quantile, or by being rejected by all women in $A$. Notice that the number of women who are matched with new partners during an iteration of the outer loop cannot exceed $\abs{\calA}$, as if $\abs{\calA}$ women did receive new partners, all active men would be matched. Therefore, the women can send at most $k \abs{\calA}$ quantile rejections (after which all active men will be rejected by all women). Similarly, the men can receive at most $k \abs{\calA}$ quantile rejections. Thus, in total the active men can witness at most $2 k \abs{\calA}$ quantile rejections. Therefore, if there are $b$ bad men after $\ell$ calls to \QuantileMatch, we must have $b \ell \leq 2 k \abs{\calA}$. Choosing $\ell = 2 \delta^{-1} k$ gives the desired result.
  \end{proof}

  We say that a man $m$ is \dft{bad in the $i$th iteration} of the outer loop in \ASM\ if $m$ became bad during the $i$th iteration and $\abs{Q^m} < 2^i$. Thus, $m \in B_i$ is bad and will not participate in any further calls to \QuantileMatch, so he will be bad when \ASM\ terminates. We denote the set bad men in the $i$ iteration by $B_i$, so that $B = B_1 \cup B_2 \cup \cdots \cup B_{\log n}$.

  \begin{lem}[Few $(2/k)$-blocking pairs]
    \label{lem:few-bad-blockers}
    Each $m \in B_i$ participates in fewer than $2^i$ $(2/k)$-blocking pairs at the termination of \ASM.
  \end{lem}
  \begin{proof}
    We will show that each bad $m \in B$ participates in at most $\abs{Q^m}$ $(2/k)$-blocking pairs, whence the lemma follows. To this end, notice that if $w \notin Q^m$, then $w$ must have rejected $m$ in some call to \QuantileMatch. Therefore, $w$ must have been matched with some $m'$ that is in the same or better quantile as $m$. By Lemma \ref{lem:w-monotonic}, when \ASM\ terminates, $w$ is still matched with someone in at least as desirable quantile as $m$, implying that $(m, w)$ is not $(2/k)$-blocking. Thus, every $(2/k)$-blocking pair $(m, w)$ must have $w \in Q^m$.
  \end{proof}

  \begin{proof}[Proof of Lemma \ref{lem:bad-guarantee}]
    Let $G_i \subseteq G$ be the set of men which are good at the termination of \ASM\ and active after the $i$th iteration of the outer loop in \ASM. Then we have $G = G_1 \cup G_2 \cup \cdots \cup G_{\log n}$. Since the number of bad men cannot increase after a call to \QuantileMatch.  By Lemma \ref{lem:few-bad-men}, if there were $b$ men which became bad in some iteration of the outer loop of \ASM, there were $\frac{1 - \delta}{\delta} b$ good men still active during the $i$th iteration. Since the number of good men can only increase in subsequent iterations, we have
    \begin{equation}
      \label{eqn:good-v-bad}
      \abs{B_i \cup B_{i+1} \cup \cdots \cup B_{\log n}} \leq b \leq \frac{\delta}{1 - \delta} \abs{G_i \cup G_{i+1} \cup \cdots \cup G_{\log n}}.
    \end{equation} 
    Applying (\ref{eqn:good-v-bad}), we can greedily form disjoint sets
    \[
    H_{\log n} \subseteq G_{\log n},\ H_{\log n-1} \subseteq G_{\log_n - 1} \cup G_{\log n}, \ldots, H_1 \subseteq G
    \]
    such that for all $i$, $H_i$ is active in the $i$th iteration and $\abs{H_i} = \frac{1 - \delta}{\delta} \abs{B_i}$. Then we compute
    \begin{align*}
      \sum_{m \in B} \abs{Q^m} &= \sum_{i = 1}^{\log n} \sum_{m \in B_i} \abs{Q^m} \leq \sum_{i = 1}^{\log n} \abs{B_i} 2^{i} \leq \sum_{i = 1}^{\log n} \frac{2 \delta}{1 - \delta} \abs{H_i} 2^{i} \leq \frac{2 \delta}{1 - \delta} \sum_{m \in G} \abs{Q^m} \leq \frac{2 \delta}{1 - \delta} \abs{E}.
    \end{align*}
    The first inequality holds by Lemma \ref{lem:few-bad-blockers}, while the second holds by the choice of the $H_i$ and the definition of the $G_i$.
  \end{proof}

  \subsection{Approximation guarantee}

  \begin{proof}[Proof of Theorem \ref{thm:approximation-guarantee}]
    By Lemma \ref{lem:good-blocking}, there are at most $4 \abs{E} / k$ blocking pairs which are not $(2/k)$-blocking. By Lemma \ref{lem:eps-blocking-good}, all $(2/k)$-blocking pairs are incident with $B$. Finally, by Lemma \ref{lem:bad-guarantee}, the bad men contribute at most $4 \delta \abs{E}$ blocking pairs for $\delta \leq 1/2$. Therefore, the total number of blocking pairs is at most $4(\delta + 1/k) \abs{E}$. Choosing $\delta = \e / 8$ and $k = \ceil{8 / \e}$ gives the desired result.
  \end{proof}

  \subsection{Run-time Guarantee}
  \label{sec:runtime}

  \begin{thm}
    \label{thm:runtime}
    \ASM$(P, \e, n)$ runs in $O(\e^{-3} \log^5(n))$ communication rounds.
  \end{thm}
  \begin{proof}
    Notice that the only communication between processors occurs in \ProposalRound. \ASM$(P, \e, n)$ iterates \QuantileMatch$(P, k)$ a total of $O(\e^{-2} \log n)$ times, while quantile match invokes \ProposalRound$(Q,k,A)$ $O(\e^{-1})$ times. Finally, each step of \ProposalRound\ can be performed in $O(1)$ communication rounds, except for Step 3, which calls \MaximalMatching. By \cite{HKP01}, \MaximalMatching\ runs in $O(\log^4 n)$ communication rounds. Thus, \ASM$(P, \e, n)$ requires $O(\e^{-3} \log^5(n))$ communication rounds, as claimed.
  \end{proof}

  \begin{rem}
    While the CONGEST model allows for unbounded local computation in each round, the local computations required by \ASM\ are quite simple. In fact, each communication round can easily be implemented in nearly-linear time in $n$. Thus the synchronous run-time of \ASM\ is $\tilde{O}(n)$. To our knowledge, this gives the first distributed algorithm whose synchronous run-time is sub-quadratic in $n$, even for unbounded preferences.
  \end{rem}

  \section{Randomized Algorithms}
  \label{sec:randomized}

  The main source of complexity in \ASM\ comes from finding a maximal matching. While Han\'n\'ckowiak, Karo\'nski, and Panconesi's algorithm \cite{HKP01} is the most efficient known deterministic algorithm, faster randomized algorithms are known. Specifically, we consider the algorithm of Israeli and Itai \cite{II86}. They describe a simple randomized distributed algorithm which finds a maximal matching in expected $O(\log n)$ rounds. By simply replacing \MaximalMatching\ with a truncated version Israeli and Itai's algorithm, we obtain a faster randomized algorithm for finding almost stable matchings. We refer the reader to Appendix \ref{sec:amm-proof} for details on the guarantees for Israeli and Itai's algorithm.

  \subsection{General preferences}

  \begin{thm}
    \label{thm:rand-asm}
    There exists a randomized distributed algorithm, \RandASM$(P, \e, n, \delta)$, which for any $\delta, \e > 0$ finds a $(1 - \e)$-stable matching with probability at least $1 - \delta$ in $O(\e^{-3} \log^2 (n / \delta \e^3))$ rounds.
  \end{thm}

  \begin{proof}[Proof sketch]
    We take \RandASM\ to be exactly the same as \ASM, except that we use Israeli and Itai's algorithm \cite{II86} for the \MaximalMatching\ subroutine. Specifically, for \MaximalMatching, we iterate \MatchingRound\ (see Appendix \ref{sec:amm-proof}) $O(\log(n / \delta \e^3))$ times. By Corollary \ref{cor:matching-probability}, each call to \MaximalMatching\ will succeed  in finding a maximal matching with probability at least $1 - O(\delta \e^3 / \log n)$. Since \RandASM\ calls \MaximalMatching\ $O(\e^{-3} / \log n)$ times, by the union bound, every call to \MaximalMatching\ succeeds with probability at least $1 - \delta$. The remaining analysis of \RandASM\ is identical to that of \ASM.
  \end{proof}

  \subsection{Almost-regular preferences}

  For $\alpha \geq 1$, We call preferences $\calP$ \dft{$\alpha$-almost-regular} if $\max_{m \in Y} \deg m \leq \alpha \min_{m \in Y} \deg m$. For example, complete preferences (where all men rank all women) are $1$-almost-regular, while uniformly bounded preferences are $\alpha$-almost-regular for $\alpha = \max_{m \in Y} \deg m$. From an algorithmic standpoint, $\alpha$-almost-regular preferences are advantageous because in order to bound the proportion of blocking edges from bad men, it suffices only to bound the number of bad men. By Lemma \ref{lem:few-bad-men}, to obtain such a guarantee, one need only iterate \QuantileMatch\ $O(1)$ rounds (instead of $O(\log n)$ times as required by \ASM).

  Further, for $\alpha$-almost-regular preferences, we can relax our requirement that \MaximalMatching\ actually find a maximal matching. We say that a player $v$ is \dft{unmatched} in $G_0$ if $v$ does not satisfy property 1 or 2 in Definition \ref{dfn:maximal-matching}. We call a subroutine \AMM$(\eta, \delta)$ which finds a matching in which only an $\eta$-fraction of players are left unmatched with probability at least $1 - \delta$ (see Appendix \ref{sec:amm-proof} for details). These unmatched players are immediately removed from play. With these simplifications, we obtain the following result.

  \begin{thm}
    \label{thm:almost-regular-asm}
    There exists a randomized distributed algorithm \AlmostRegularASM$(P, \e, \delta, \alpha)$ which for $\alpha$-almost-regular preferences $P$ finds a $(1 - \e)$-stable matching with probability at least $1 - \delta$. The run-time of \AlmostRegularASM$(P, \e, \delta, \alpha)$ is $O(\alpha \e^3 \log(\alpha/\delta \e))$ rounds.
  \end{thm}
  \begin{proof}[Proof sketch]
    \AlmostRegularASM$(P, \e, \delta, \alpha)$ works by iterating \QuantileMatch\ $O(\alpha \e^{-2})$ times, which by Lemma \ref{lem:few-bad-men} implies that only $\e / 4 \alpha$ fraction of men are bad. 

    We modify \ProposalRound\ to call \AMM$(\eta, \delta')$ instead of \MaximalMatching. \AMM\ runs in $O(\log((\eta \delta')^{-1}))$ and finds a $(1 - \eta)$-maximal matching with probability $(1 - \delta')$. Since \AMM\ is called $O(\alpha \e^{-3})$ times, choosing $\eta = O(\e^4 / \alpha)$ and $\delta' = O(\delta \e^3 / \alpha$), \AMM\ will leave at most an $\e / 4 \alpha$ fraction of men unmatched in any call \AMM\ with probability at least $1 - \delta$, by the union bound. Such unmatched men are immediately removed from play.

    By the preceding two paragraphs, \AlmostRegularASM\ produces a matching in which at most an $\e / 2 \alpha$ fraction of men are either bad or unmatched. By $\alpha$-almost-regularity, these men can contribute at most $\frac{\e}{2} \abs{E}$ blocking pairs. The remaining men are good, and therefore by Lemmas \ref{lem:eps-blocking-good} and \ref{lem:good-blocking}\footnote{Although these lemmas were proven assuming that \MaximalMatching\ found a maximal matching (not an almost maximal matching) the proofs remain valid as long as the small fraction of unmatched players immediately remove themselves from play.} contribute at most $\frac{\e}{2} \abs{E}$ blocking pairs.
  \end{proof}



  \bibliography{asm-alg}{}
  \bibliographystyle{plain}

  \appendix


  \section{Randomized maximal and almost maximal matchings}
  \label{sec:amm-proof}

  Israeli and Itai's \cite{II86} algorithm for finding a maximal matching works by identifying a sparse subgraph of $\commgraph$, then finding a large matching $\matching_1$ in the sparse subgraph. The edges and incident vertices of $\matching_1$, as well as remaining isolated vertices, are removed from $\commgraph$ resulting in a subgraph $\commgraph_1$. The process is iterated, giving a sequence of subgraphs $\commgraph_1, \commgraph_2, \ldots$ and matchings $\matching_1, \matching_2, \ldots$, until $\commgraph_k = \varnothing$. At this point, $\matching = \bigcup_{i = 1}^k \matching_i$ is a maximal matching. We give pseudocode for Israeli and Itai's main subroutine, which we call \MatchingRound, in Algorithm \ref{alg:matching-round}. In \cite{II86}, Israeli and Itai prove the following performance guarantee for \MatchingRound.

  \begin{algorithm}
    \caption{\MatchingRound$(\commgraph)$: Finds a large matching in a graph}
    \label{alg:matching-round}
    \begin{algorithmic}[1]
      \STATE Each $v \in V$ picks a neighbor $w$ uniformly at random, forms oriented edge $(v,w)$.
      \STATE Each $v \in V$ with $\deg_{in}(v) > 0$ picks one in-coming edge $(w, v)$ uniformly at random, deletes remaining in-edges. Let $G'$ be the (undirected) graph formed by the chosen edges with orientation ignored.
      \STATE Each $v \in V$ with $\deg_{\commgraph'}(v) > 0$ chooses one incident edge $(v, w)$ uniformly at random.
      \STATE The matching $M_1$ consists of edges $(v, w) \in \commgraph'$ which were chosen by both $v$ and $w$ in the previous round. $\commgraph_1 = (V_1, E_1)$ is the induced subgraph of $\commgraph$ formed by removing all vertices contained in $M_1$ and any remaining isolated vertices from $\commgraph$.
      \STATE Output $(G_1, M_1)$.
    \end{algorithmic}
  \end{algorithm}

  \begin{lem}(Israeli and Itai \cite{II86})
    \label{lem:matching}
    There exists an absolute constant $c < 1$ such that on input $\commgraph = \commgraph_0 = (V_0, E_0)$, the resulting graph $\commgraph_1 = (V_1, E_1)$ found by \MatchingRound\ satisfies $\E(\abs{V_1}) \leq c \abs{V_0}$.
  \end{lem}

  As a consequence of Lemma \ref{lem:matching}, we obtain the following useful result.

  \begin{cor}
    \label{cor:matching-probability}
    Let $\eta > 0$ be a parameter. Then $s = O(\log(n / \eta))$ iterations of \MatchingRound\ suffice to produce a maximal matching in $\commgraph$ with probability at least $1 - \eta$.
  \end{cor}
  \begin{proof}
    By Lemma \ref{lem:matching}, we have $\E(\abs{V_s}) \leq c^s n$. Therefore, applying Markov's inequality gives
    \[
    \Pr(\abs{V_s} \geq 1) \leq \frac{\E(\abs{V_s})}{1} \leq c^s n.
    \]
    The result follows by taking $s \geq \log(n / \eta) / \log(c^{-1})$.
  \end{proof}

  The almost regular variant of \ASM\ only requires a subroutine that we finds matchings which are almost maximal.

  \begin{dfn}
    \label{dfn:amm}
    Let $\commgraph = (V, E)$ be a communication graph and $M \subset E$ a matching in $G$. For $0 < \eta \leq 1$, we say that $M$ is \dft{$(1 - \eta)$-maximal} if the set $V'$ of vertices not satisfying conditions 1 or 2 in Definition \ref{dfn:maximal-matching} satisfies $\abs{V'} \leq \eta \abs{V}$.
  \end{dfn}

 We can apply Lemma \ref{lem:matching} to give a constant round algorithm which finds almost maximal matchings.

  \begin{cor}
    \label{cor:amm}
    There exists a randomized distributed algorithm \AMM$(\commgraph, \eta, \delta)$ which finds a $(1 - \eta)$-maximal matching with probability at least $(1 - \delta)$. \AMM$(\commgraph, \eta, \delta)$ runs in $O(\log(\eta^{-1} \delta^{-1}))$ rounds.
  \end{cor}
  \begin{proof}
    Consider the algorithm which iterates \MatchingRound\ $s$ times. We apply Lemma \ref{lem:matching} and Markov's inequality to obtain
    \[
    \Pr(\abs{V_s} \geq \eta n) \leq \frac{c^s n}{\eta n} = \eta^{-1} c^s.
    \]
    Choosing $s = O(\log(\delta^{-1} \eta^{-1}))$, we have $\eta^{-1} c^s \leq \delta$, which gives the desired result.
  \end{proof}



\end{document}